\documentclass{IEEEtran}
\usepackage{amsmath,amssymb,amsfonts}
\usepackage{amsthm}
\usepackage{algorithm}
\usepackage{algpseudocodex}
\usepackage{graphicx}
\usepackage{textcomp}

\usepackage{caption}

\usepackage[colorlinks=true,
            linkcolor=blue,
            citecolor=blue,
            urlcolor=blue]{hyperref}

\newcommand{\lf}{\left}
\newcommand{\rg}{\right}
\DeclareMathOperator*{\argmax}{\arg\max}

\DeclareMathOperator*{\lat}{latt}

\theoremstyle{plain}
\newtheorem{theorem}{Theorem}
\newtheorem{lemma}{Lemma}
\newtheorem{corollary}{Corollary}

\theoremstyle{definition}
\newtheorem{definition}{Definition}
\newtheorem{assumption}{Assumption}
\newtheorem{example}{Counter-Example}

\theoremstyle{remark}
\newtheorem{remark}{Remark}

\renewcommand{\phi}{\varphi}
\renewcommand{\epsilon}{\varepsilon}

\newcommand{\R}{\mathbb{R}}
\newcommand{\extended}{\bar{\R}}
\newcommand{\N}{\mathbb{N}}
\newcommand{\Rn}{\R^{n}}
\newcommand{\Ru}{\R^{p_1}}
\newcommand{\Rd}{\R^{p_2}}

\newcommand{\sbs}{\subseteq}
\newcommand{\ssbs}{\subset}

\newcommand{\tms}{\times}
\newcommand{\ii}{\infty}

\newcommand{\dynamics}{f}

\newcommand{\tVal}{t}
\newcommand{\tNot}{\tVal_0}
\newcommand{\tDum}{s}
\newcommand{\tDumDum}{\tau}
\newcommand{\tInit}{S}
\newcommand{\tFin}{T}

\newcommand{\xVal}{x}
\newcommand{\xDum}{y}

\newcommand{\xNot}{\xVal_0}
\newcommand{\uVal}{u}
\newcommand{\dVal}{d}

\newcommand{\uVals}{\mathcal{U}}
\newcommand{\dVals}{\mathcal{D}}

\newcommand{\uSig}{\mathrm{u}}
\newcommand{\dSig}{\mathrm{d}}

\newcommand{\uSigs}{\mathrm{U}}
\newcommand{\dSigs}{\mathrm{D}}

\newcommand{\dStrat}{\delta}
\newcommand{\dStrats}{\Delta}

\newcommand{\xSig}{\mathrm{x}}
\newcommand{\xSigs}{\mathrm{X}}

\newcommand{\traj}[2]{\xSig_{#1}^{#2}}
\newcommand{\naughtTraj}{\traj{x_0,t_0}{\uSig,\dSig}}
\newcommand{\naughtTrajDerivative}{\dot{\xSig}_{x_0,t_0}^{\uSig,\dSig}}

\newcommand{\stratTraj}{\traj{x,t}{\uSig,\dStrat(\uSig)}}

\newcommand{\evalTimes}{\mathcal{T}}
\newcommand{\rob}{J}

\newcommand{\supu}{\sup_{\uSig \in \uSigs}}
\newcommand{\infd}{\inf_{\dStrat \in \dStrats}}
\newcommand{\val}{V}

\newcommand{\robRA}{J_\mathrm{RA}}
\newcommand{\robR}{J_\mathrm{R}}
\newcommand{\robA}{J_\mathrm{A}}
\newcommand{\valRA}{\val_\mathrm{RA}}
\newcommand{\valR}{\val_\mathrm{R}}
\newcommand{\valA}{\val_\mathrm{A}}

\newcommand{\tar}{r}
\newcommand{\tarSpec}{\tilde{r}}
\newcommand{\con}{q}

\newcommand{\tarEval}[2]{\tar \lf( #1(#2), #2 \rg)}
\newcommand{\tariEval}[2]{\tar_i \lf( #1(#2), #2 \rg)}

\newcommand{\tarStarEval}[2]{\tar^* \lf( #1(#2), #2 \rg)}
\newcommand{\tarSpecEval}[2]{\tarSpec \lf( #1(#2), #2 \rg)}
\newcommand{\conEval}[2]{\con \lf( #1(#2), #2 \rg)}

\newcommand{\bargoala}{\bar{\mathcal{G}}_{1,a}}
\newcommand{\goala}{\mathcal{G}_{1,a}}
\newcommand{\goalb}{\mathcal{G}_{1,b}}
\newcommand{\goalone}{\mathcal{G}_{1}}
\newcommand{\bargoaltwo}{\bar{\mathcal{G}}_{2}}
\newcommand{\goaltwo}{\mathcal{G}_{2}}
\newcommand{\obstacleone}{\mathcal{O}_1}
\newcommand{\obstacletwo}{\mathcal{O}_2}

\newcommand{\robSRA}{\rob_{\mathrm{SRA}}}
\newcommand{\valSRA}{\val_{\mathrm{SRA}}}

\newcommand{\robURA}{\rob_{\mathrm{URA}}}
\newcommand{\valURA}{\val_{\mathrm{URA}}}
\newcommand{\tarFam}{\mathcal{R}}

\newcommand{\robRAA}{\rob_{\mathrm{RAA}}}
\newcommand{\valRAA}{\val_{\mathrm{RAA}}}

\newcommand{\lpf}{p}

\newcommand{\compact}{\mathcal{K}}
\newcommand{\RleT}{\R_{\le \tFin}}
\newcommand{\blackbox}{\hfill\rule{6pt}{6pt}}
\newcommand{\homeo}{\sigma}

\title{Exact Decomposition of Value Functions for Two-Player Games in Hamilton-Jacobi Reachability
\author{Dylan Hirsch*, William Sharpless*, and Sylvia Herbert}
\thanks{Research reported in this publication was supported by the National Institutes of Health under award number T32EB009380. The content is solely the responsibility of the authors and does not necessarily represent the official views of the National Institutes of Health.}
\thanks{Dylan Hirsch (corresponding author), William Sharpless, and Sylvia Herbert are with the Department of Mechanical and Aerospace Engineering, University of California, San Diego, 9500 Gilman Drive, La Jolla, CA 92093.
        {\tt\small dhirsch@ucsd.edu,
        wsharpless@ucsd.edu,
        sherbert@ucsd.edu.}}%
}

\begin{document}
\maketitle

\begin{abstract}
    Hamilton-Jacobi reachability (HJR) is an important framework in safe control theory.
    HJR provides theoretical tools and numerical methods to obtain value-functions for tasks involving target-reaching and obstacle-avoidance.
    A recent work proposed an algebraic framework to scale to complex, composite tasks via decomposing the value functions of the composite tasks into value functions for the fundamental tasks traditionally studied in HJR, all in the one-player setting.
    Many of these decomposition results, however, do not directly translate to the two-player setting.
    In this technical note, we nevertheless show that one of the previous work's key results for analyzing various tasks involving reaching multiple targets while obeying constraints still holds in the two-player setting, so long as a critical monotonicity assumption is satisfied.
    In particular, we detail the analogous result and its proof (which structurally differs from the one-player case), we show via a counter-example what can go wrong if this monotonicity assumption is not satisfied, and we show how the analysis of a variety of tasks can be performed using this result.
    Such tasks include sequential and any-order multi-target-reaching tasks under constraints.
    We note that, whereas the prior work considered discrete-time, infinite-horizon tasks, which are more standard in reinforcement learning, we here consider continuous-time, finite-horizon tasks, which are more standard in HJR.
\end{abstract}

\section{Introduction}
Hamilton-Jacobi reachability (HJR) is a powerful framework for safety-critical control that is grounded in rigorous theory, complemented by efficient numerical tools \cite{mitchell-2005,lygeros-2011,fisac-chen-2015,HJR-Survey}.
In brief, HJR builds on the theory of differential games to compute value functions for tasks involving goal-reaching and obstacle-avoidance.
While the early works in HJR studied fundamental tasks, including the reach, avoid, and reach-avoid tasks, more recent works have focused on computing value functions for more complex, composite tasks \cite{stl-meets-reachability,sharpless2026dual,sharpless2026bellman,Xiang-CDC-2025,hirsch-gra}.
In particular, building on the work in \cite{sharpless2026dual}, the recent work \cite{sharpless2026bellman} provides an algebraic approach to decomposing value functions of composite tasks into the value functions of HJR's fundamental tasks when in the one-player setting.

One of the fundamental decomposition formulae in this work allows one to break the value-function computation for certain composite tasks with upstream and downstream components into a three-step process: (1) compute the value function for the downstream process by any means, (2) form a new time-varying target function using the downstream value function, and (3) compute a reach-avoid value function with this new target function and the constraint function from the upstream task.

This rule, however, does not generally apply in the two-player setting, i.e. in the presence of an adversarial disturbance.
In this technical note, we provide the proper analogue for the two-player setting.
We note that in this work, we also consider the continuous-time, finite-horizon setting, which is more traditional in HJR, whereas the previous work considered a discrete-time, infinite-horizon setting, which is more traditional in reinforcement learning.

The critical difference between these decomposition results in the one-player and two-player settings is a certain monotonicity assumption on the downstream task that is required in the latter.
We demonstrate via a counter-example why the decomposition fails in the two-player setting without this assumption.

We also show that many of the decomposition results previously derived in the literature can be recovered (and even extended) using this single formula.
For example, decompositions for the two-player reach always-avoid and reach-reach problems were studied in \cite{hirsch-2206}, and sequential reach-avoid problems were studied in \cite{Xiang-CDC-2025,chen2025controlsynthesismultiplereachavoid}.

~\\
\noindent\textbf{Notation:}
We denote the extended real line by $\extended := [-\ii,\ii]$, having its usual topology.
For each $\tFin \in \R$ we let $\RleT := (-\ii,\tFin]$.
We let $\N$ be the natural numbers (not including $0$), and for each $N \in \N$ we let $[N] = \{1,\dots,N\}$.
Given a set $\mathcal{S}$ and some $i \in \N$, we denote the set of all subsets of $\mathcal{S}$ with exactly $i$ elements by $\binom{\mathcal{S}}{i}$.

\section{Setup and Background}\label{sec:background}
The setup and background for this article is almost identical to the general HJR setup described in \cite{hirsch-gra}.
We reproduce this information here for completeness.

\subsection{Dynamics}
We consider a system with dynamics
\begin{equation}\label{eqn:dynamics}
\dot{\xSig}(\tVal) = \dynamics(\xSig(\tVal), \uSig(\tVal), \dSig(\tVal), \tVal),
\end{equation}
where $\dynamics:\Rn \tms \uVals \tms \dVals \tms \R \to \Rn$, with $\uVals \ssbs \Ru$ and $\dVals \ssbs \Rd$.
Here, $\xSig$ is the state, $\uSig$ is the control, and $\dSig$ is the disturbance.
We will assume the following throughout the sequel, which guarantee global existence and uniqueness of trajectories for this system.

\begin{assumption}\label{assumption:compactness}
    The sets $\uVals$ and $\dVals$ are compact.
\end{assumption}
\begin{assumption}\label{assumption:regularity}
~
\begin{itemize}
    \item For all $\tVal \in \R$, the map $(\xVal, \uVal, \dVal) \mapsto \dynamics(\xVal, \uVal, \dVal, \tVal)$ is continuous.
    \item For all $\xVal \in \Rn$, $\uVal \in \uVals$, and $\dVal \in \dVals$, the map $\tVal \mapsto \dynamics(\xVal, \uVal, \dVal, \tVal)$ is measurable.
    \item There exists some $K > 0$ such that $\|\dynamics(\xVal,\uVal,\dVal,\tVal)\| \le K(1 + \|\xVal\|)$ for all $\uVal \in \uVals$, $\dVal \in \dVals$, and $\tVal \in \R$.
    \item For each compact set $\compact \ssbs \Rn$ there exists an $L > 0$ such that $\|\dynamics(\xVal_1,\uVal,\dVal,\tVal) - \dynamics(\xVal_2,\uVal,\dVal,\tVal)\| \le L \|\xVal_1 - \xVal_2\|$ for all $\xVal_1,\xVal_2 \in \compact$, $\uVal \in \uVals$, $\dVal \in \dVals$, and $\tVal \in \R$.
\end{itemize}
\end{assumption}

\subsection{Signals and trajectories}
We denote by $\uSigs$ the set of all measurable control signals $\uSig: \R \to \uVals$ and by $\dSigs$ the set of all measurable disturbance signals $\dSig: \R \to \dVals$.
We now precisely define the state trajectory $\naughtTraj: \R \to \Rn$ that results from a control signal $\uSig$ and disturbance signal $\dSig$, given that the system is in state $\xNot$ at time $\tNot$.

For each $\xNot \in \Rn$, $\tNot \in \R$, $\uSig \in \uSigs$, and $\dSig \in \dSigs$,
we let $\naughtTraj$ be the Carath\'{e}odory solution of \eqref{eqn:dynamics} under the condition $\xSig(\tNot) = \xNot$.
More explicitly, $\naughtTraj$ is defined to be the unique locally absolutely continuous map from $\R$ to $\Rn$ for which $\naughtTraj(\tNot) = \xNot$ and $\naughtTrajDerivative(\tVal) = \dynamics( \naughtTraj(\tVal), \uSig(\tVal), \dSig(\tVal), \tVal)$ for a.e. $\tVal \in \R$ (existence and uniqueness of this solution follow from Assumptions \ref{assumption:compactness} and \ref{assumption:regularity}; see Theorem 1.2.1 in \cite{Friedman-Differential-Games}). %\cite{Friedman-Differential-Games}

\subsection{Performance functionals}

We will denote by $\xSigs$ the set of all plausible trajectories, i.e. continuous maps $\xSig: \R \to \Rn$.
We endow $\xSigs$ with the topology of uniform convergence on compact sets.
In this topology, a sequence $(\xSig_{i})_{i \in \N}$ in $\xSigs$ converges to $\xSig \in \xSigs$ iff $\max_{\tVal \in \compact} \| \xSig_i(\tVal) - \xSig(\tVal) \| \to 0$ for all non-empty, compact $\compact \ssbs \R$.

We can think of any map $\rob: \xSigs \tms \evalTimes \to \extended$, where $\evalTimes \sbs \R$, as a performance functional for some task, where $\rob$ assigns a score $\rob(\xSig, \tVal)$ to a trajectory $\xSig$ based on the time $\tVal$ at which the task begins.
Here, $\evalTimes$ represents the set of times at which it is reasonable for the task to begin (e.g. for the task ``remain inside until noon today," we may define $\evalTimes$ to be times no later than noon).

In HJR, the performance functional $\rob$ for each task is chosen such that $\rob(\xSig, \tVal) > 0$ iff the trajectory $\xSig$ satisfies the task starting at time $\tVal$.\footnote{More generally, the robustness metric for any specification in temporal logic is a performance functional whose sign corresponds to qualitative satisfaction of the specification \cite{donze-robustness-metric}.}
Three key examples are the reach, avoid, and reach-avoid performance functionals.
For each $\tFin \in \R$ and continuous $\tar, \con: \Rn \tms \RleT \to \R$, we define $\robR[\tar; \tFin], \robA[\con;\tFin], \robRA[\tar,\con;\tFin]: \xSigs \tms \RleT \to \R$ by
\begin{align*}
    \robR[\tar;\tFin](\xSig, \tVal) &= \max_{\tDum \in [\tVal, \tFin]} \tarEval{\xSig}{\tDum},\\
    \robA[\con;\tFin](\xSig, \tVal) &= \min_{\tDumDum \in [\tVal, \tFin]} \conEval{\xSig}{\tDumDum},\\
    \robRA[\tar,\con;\tFin](\xSig, \tVal) &= \max_{\tDum \in [\tVal, \tFin]} \min\{ \tarEval{\xSig}{\tDum}, \min_{\tDumDum \in [\tVal, \tDum]} \conEval{\xSig}{\tDumDum} \}. 
\end{align*}

We will typically require the following technical condition of our performance functionals:
\begin{definition}[Definition 4 in \cite{hirsch-gra}]
Let $\evalTimes \sbs \R$.
A performance functional $\rob: \xSigs \tms \evalTimes \to \extended$ is \textbf{past-independent} if for all $\tVal \in \evalTimes$ and all $\xSig_1, \xSig_2 \in \xSigs$, if $\xSig_1(\tDum) = \xSig_2(\tDum)$ for every $\tDum \ge \tVal$, then $\rob(\xSig_1, \tVal) = \rob(\xSig_2, \tVal)$.
\end{definition}
Note that the reach, avoid, and reach-avoid performance functionals discussed above are all continuous and past-independent (Lemma 3 and Observation 2 in \cite{hirsch-gra}).

\subsection{Differential games and value functions}
We now consider a differential game in which the controller player would like to maximize a given performance functional and the disturbance player would like to minimize it.

We first specify the information pattern for the game.
Conceptually, we can think of any map $\dStrat:\uSigs \to \dSigs$ as representing a strategy for a disturbance player that selects the disturbance input $\dSig = \dStrat(\uSig)$ based upon the control input $\uSig$.
In HJR, we make the restriction that the disturbance's strategy cannot use future information to inform current decisions via the following definition.

\begin{definition}[Non-Anticipativity]
The map $\dStrat :\uSigs \to \dSigs$ is \textbf{non-anticipative} if for all $\tVal \in \R$ and $\uSig_1,\uSig_2 \in \uSigs$ such that $\uSig_1(\tDum) = \uSig_2(\tDum)$ for a.e. $\tDum \le \tVal$, we also have $\dStrat[\uSig_1](\tDum) = \dStrat[\uSig_2](\tDum)$ for a.e. $\tDum \le \tVal$.
\end{definition}
\noindent We denote by $\dStrats$ the set of all non-anticipative $\dStrat: \uSigs \to \dSigs$.

Adapting the notation used in \cite{sharpless2026bellman}, given a performance functional $\rob: \xSigs \tms \evalTimes \to \extended$, we define the value function $\val[\rob]: \Rn \tms \evalTimes \to \extended$ by
\begin{equation}\label{eqn:value-function-definition}
    \val[\rob](\xVal, \tVal) = \infd \supu \rob(\stratTraj, \tVal).
\end{equation}

The value function $\val[\rob](\xVal, \tVal)$ conceptually represents the performance score that will result from both players acting optimally under the non-anticipative information structure, given that the system is initialized in state $\xVal$ at time $\tVal$.
We have the following useful fact.
\begin{lemma}[Lemma 1 in \cite{hirsch-gra}] \label{lem:continuous-evaluations-have-continuous-values}
Let $\rob: \xSigs \tms \evalTimes \to \extended$ for some $\evalTimes \sbs \R$.
If $\rob$ is continuous, then $\val[\rob]$ is also continuous.
If we in addition have $\rob < \ii$ (resp. $\rob > -\ii$), then $\val[\rob] < \ii$ (resp. $\val[\rob] > -\ii$). 

\end{lemma}

In HJR, the sign of the value $\val[\rob](\xVal, \tVal)$ is used to determine whether a task can be satisfied starting from the state $\xVal$ at time $\tVal$.
Formally, we have the following result:
\begin{lemma}[Lemma 2 in \cite{hirsch-gra}]\label{lem:value-and-satisfiability}
Let $\evalTimes \sbs \R$, $\rob: \xSigs \tms \evalTimes \to \extended$, $\xVal \in \Rn$, and $\tVal \in \evalTimes$.
If $\val[\rob](\xVal, \tVal) > 0$, then for each $\dStrat \in \dStrats$ there is a $\uSig \in \uSigs$ such that $\rob(\stratTraj, \tVal) > 0$.
Similarly, if $\val[\rob](\xVal, \tVal) < 0$, then there is some $\dStrat \in \dStrats$ such that $\rob(\stratTraj, \tVal) < 0$ for all $\uSig \in \uSigs$.
\end{lemma}
For each $\tFin \in \R$ and continuous $\tar,\con: \Rn \tms \RleT \to \R$, we will for convenience let
\begin{align*}
    \valR[\tar;\tFin] &= \val[\robR[\tar;\tFin]],\\
    \valA[\con;\tFin] &= \val[\robA[\con;\tFin]],\\
    \valRA[\tar,\con;\tFin] &= \val[\robRA[\tar,\con;\tFin]].
\end{align*}
It follows from Lemma \ref{lem:continuous-evaluations-have-continuous-values} that each of the above value functions is continuous.

\section{Decomposition Theorems}

\subsection{Tasks involving an upstream reach-avoid task and a general downstream task}

We here present the main contribution of this technical note, the analogue of Lemma 1 in \cite{sharpless2026bellman} to the two-player, continuous-time, finite-horizon setting.
The critical difference compared with the previous result is the monotonicity assumption \eqref{eqn:montonicity-condition}, which is uniquely needed in the two-player setting.

Before stating the theorem, let us briefly and informally foreshadow the setup and the monotonicity assumption.
We will consider the performance functional corresponding to a task specification in which the system must first reach a goal, while obeying constraints in the process, and then subsequently accomplish a downstream task.
Letting the function $\con$ encode the constraint and the performance functional $\rob_0$ encode the downstream task, we require that
\begin{align}\label{eqn:montonicity-condition}
    \min\{\min_{\tDumDum \in [\tVal, \tDum_1]}&\conEval{\xSig}{\tDumDum}, \rob_0(\xSig, \tDum_1)\} \nonumber\\
    &\ge 
    \min\{\min_{\tDumDum \in [\tVal, \tDum_2]}\conEval{\xSig}{\tDumDum}, \rob_0(\xSig, \tDum_2)\}
\end{align}
whenever $\tVal \le \tDum_1 \le \tDum_2 \le \tFin$ (here $\tFin$ is the horizon time of the task).
Conceptually, this requirement states that the task ``obey the constraints until time $\tDum$ and subsequently perform the downstream task'' should not grow easier as the intermediate time $\tDum$ grows later.

This requirement holds, for example, when the downstream task specified by $\rob_0$ is to reach some collection of goals before time $\tFin$ (see Corollaries \ref{cor:ordered-reach-avoid} and \ref{cor:unordered-reach-avoid}).
Indeed, as the intermediate time $\tDum$ becomes larger, the task ``obey the constraints until time $\tDum$ and subsequently reach all the goals'' becomes more challenging (or at least no less challenging) because one has less time to reach the goals.

\begin{theorem}\label{thm:main-theorem}
    Let $\tFin \in \R$, and let $\tar,\con: \Rn \tms \RleT \to \R$ be continuous.
    Let $\rob_0: \xSigs \tms \RleT \to (-\ii, \ii]$ be continuous and past-independent.
    Define $\rob: \xSigs \tms \RleT \to \R$ by
    \begin{equation}\label{eqn:main-theorem-composite-performance-funcitonal}
        \rob(\xSig, \tVal) = \max_{\tDum \in [\tVal, \tFin]} \min\{\tarEval{\xSig}{\tDum}, \min_{\tDumDum \in [\tVal, \tDum]} \conEval{\xSig}{\tDumDum},  \rob_0(\xSig, \tDum)\}.
    \end{equation}
    Then $\rob$ is continuous and past-independent.
    If, in addition, the monotonicity condition \eqref{eqn:montonicity-condition} is satisfied for each $\xSig \in \xSigs$ and each $\tVal,\tDum_1,\tDum_2 \in \RleT$ for which $\tVal \le \tDum_1 \le \tDum_2$, then
    \begin{equation}\label{eqn:main-theorem-conclusion}
        \val[\rob] = \valRA[\tarSpec, \con; \tFin],
    \end{equation}
    where $\tarSpec := \min\{\tar, \val[\rob_0]\}$.
\end{theorem}
\begin{proof}
    See Section \ref{sec:appendix-main-theorem} in the appendix.
\end{proof}

Let us briefly explain the significance of this result.
We consider a composite task composed of an upstream reach-avoid task with performance functional $\robRA[\tar, \con; \tFin]$ and a downstream task with performance functional $\rob_0$.
This downstream task can be relatively general, with the only requirements being that its performance functional is continuous, past-independent, and defined on $\xSigs \tms \RleT$.

Given an initial time $\tVal \le \tFin$, we would like to calculate the performance functional of a composite task in which the upstream reach-avoid task must first be completed at any intermediate time $\tDum \in [\tVal, \tFin]$ and the downstream task must be completed thereafter.
More precisely, the performance functional $\rob$ for this composite task is defined by \eqref{eqn:main-theorem-composite-performance-funcitonal}.

This theorem states that, so long as the monotonicity condition \eqref{eqn:montonicity-condition} holds, to compute the value function for this composite task, one can proceed in three steps: (1) the value function $\val[\rob_0]$ for the downstream task is computed by any means, (2) the new target function $\tarSpec := \min\{\tar, \val[\rob_0]\}$ is formed, and (3) the reach-avoid value function $\valRA[\tarSpec, \con; \tFin]$ is computed (which can be done by solving a Hamilton-Jacobi-Isaacs partial differential equation \cite{fisac-chen-2015}).
This result guarantees that this last value function will be precisely $\val[\rob]$, the value function for the composite task.

In the two-player case, the result generally may not hold if the monotonicity condition \eqref{eqn:montonicity-condition} is not satisfied.
An example demonstrating how the Theorem can fail when this condition does not hold is presented in Counter-Example \ref{ex:counter-example} in Section \ref{sec:counter-example}.

First, we show how this theorem can be applied.
In particular we recover several separate decomposition results from the literature.

\subsection{Tasks with more complex structures}

At first glance, it appears that Theorem \ref{thm:main-theorem} only allows one to decompose composite tasks in which events must occur in a prescribed order, such as the sequential reach-avoid task in \cite{Xiang-CDC-2025} or the reach-always-avoid task in \cite{hirsch-2206}.
In fact, however, this same theorem can be leveraged to decompose tasks with rather complex structures.
As an important example of this claim, we use Theorem \ref{thm:main-theorem} to prove a generalized version of itself, which gives a method to decompose value functions for composite tasks in which success depends on several parallel sub-tasks, where each sub-task is of the form in Theorem \ref{thm:main-theorem}.
In Section \ref{sec:unordered-reach-avoid}, we will use this theorem to produce an algorithm for a task in which $N$ targets must be achieved in any order.

We will use the concept of a lattice term function to logically join together the types of tasks considered in the previous subsection into a more complex composite task.
\begin{definition}[Lattice term function]
Let $N \in \N$ and $\lpf: \extended^N \to \extended$.
\begin{itemize}
    \item $\lpf$ is a \textbf{$0$-ary lattice term function} if it is a coordinate map, i.e. there is some $i \in [N]$ such that $\lpf(y) = y_i$ for all $y \in \extended^N$,
    \item For each $k \in \N$, $\lpf$ is a \textbf{$k$-ary lattice term function} if $\lpf = \min\{\lpf_1, \lpf_2\}$ or $\lpf = \max\{\lpf_1, \lpf_2\}$ for some $(k-1)$-ary lattice term functions $\lpf_1$ and $\lpf_2$,
    \item $\lpf$ is a \textbf{lattice term function} if it is a $k$-ary lattice term function for some $k \in \{0\} \cup \N$. 
\end{itemize}
\end{definition}

We now will consider several tasks of the form previously studied, each involving a target function $\tar_i$, a performance functional $\rob_{0,i}$ for a downstream task, and a shared constraint function $\con$.
For each $\rob_{0,i}$ we require that the monotonicity assumption is satisfied:
\begin{align}\label{eqn:montonicity-condition-i}
    \min\{\min_{\tDumDum \in [\tVal, \tDum_1]}&\conEval{\xSig}{\tDumDum}, \rob_{0,i}(\xSig, \tDum_1)\} \nonumber\\
    &\ge 
    \min\{\min_{\tDumDum \in [\tVal, \tDum_2]}\conEval{\xSig}{\tDumDum}, \rob_{0,i}(\xSig, \tDum_2)\}
\end{align}
whenever $\tVal \le \tDum_1 \le \tDum_2 \le \tFin$.

\begin{theorem}\label{thm:main-theorem-2}
    Let $\tFin \in \R$ and $N \in \N$.
    Let $\tar_1,\dots,\tar_N, \con: \Rn \tms \RleT \to \R$ be continuous.
    Let $\rob_{0,1},\dots,\rob_{0,N}: \xSigs \tms \RleT \to (-\ii,+\ii]$ be continuous and past-independent.
    For each $i = 1,\dots,N$, define $\rob_{i}: \xSigs \tms \RleT \to \R$ by
    \begin{align*}
        &\rob_{i}(\xSig, \tVal)\\ &=\max_{\tDum \in [\tVal, \tFin]} \min\{\tariEval{\xSig}{\tDum}, \min_{\tDumDum \in [\tVal, \tDum]} \conEval{\xSig}{\tDumDum},  \rob_{0,i}(\xSig, \tDum)\}.
    \end{align*}
    Let $\lpf:\extended^N \to \extended$ be a lattice term function, and define $\rob = \lpf \circ (J_{1},\dots,J_{N})$.
    Then $\rob$ is continuous and past-independent.
    If, in addition, the monotonicity condition \eqref{eqn:montonicity-condition-i} is satisfied for each $i \in [N]$, $\xSig \in \xSigs$, and $\tVal,\tDum_1,\tDum_2 \in \RleT$ for which $\tVal \le \tDum_1 \le \tDum_2$, then
    \begin{equation}\label{eqn:main-theorem-2-conclusion}
        \val[\rob] = \valRA[\tarSpec, \con; \tFin],
    \end{equation}
    where 
    \begin{align*}
    & \tarSpec := \max_{i \in [N]} \min\{\tar_i, \val[\lpf \circ (\rob_1, \dots, \rob_{i-1}, \rob_{0,i}, \rob_{i+1}, \dots, \rob_N)]\}.
    \end{align*}
\end{theorem}
\begin{proof}
    See Section \ref{sec:appendix-main-theorem-2} in the appendix.
\end{proof}    

\section{Using the decomposition theorems to compute the value function for various tasks}
We now show how the relatively general decomposition formulae in Theorems \ref{thm:main-theorem} and \ref{thm:main-theorem-2} can be used to easily derive specific decomposition formulae for various tasks.
In particular we will consider:
\begin{enumerate}
    \item[(1)] a sequential reach-avoid task, where the constraints can change after a target is reached,
    \item[(2)] an unordered reach-avoid task, where the constraints do not switch after a target is reached,
    \item[(3)] a reach always-avoid task, in which a system must not only remain within constraints until a target is reached (as in the normal reach-avoid task), but also thereafter.
\end{enumerate}
Task (1) was considered in \cite{Xiang-CDC-2025,chen2025controlsynthesismultiplereachavoid} and tasks (2) and (3) were considered in \cite{hirsch-2206}.
We note that for task (2), the prior work only considered two targets and no constraint.

\subsection{The sequential reach-avoid task}\label{sec:sequential-reach-avoid}

We now consider a sequential reach-avoid task in which a system must reach a sequence of targets in a prescribed order, while also satisfying constraints.
We allow the active constraints to switch once a target is reached.

More precisely, given a terminal time $\tFin \in \R$, continuous target functions $\tar_1, \dots, \tar_N: \Rn \tms \RleT \to \R$, and continuous constraint functions $\con_1, \dots, \con_N: \Rn \tms \RleT \to \R$, we define the sequential reach-avoid performance functional $\robSRA[\tar_1,\con_1,\dots,\tar_N,\con_N; \tFin]: \xSigs \tms \RleT \to \R$ by
\begin{align*}
    \robSRA[\tar_1,&\con_1,\dots,\tar_N,\con_N; \tFin](\xSig, \tVal) = \max_{\tVal = \tDum_0 \le \tDum_1 \le \tDum_{2} \le \dots \le \tDum_{N} \le \tFin} \\
    &\min_{i=1,\dots,N}
    \min\{\tar_i(\xSig(\tDum_i), \tDum_i), \min_{\tDumDum \in [\tDum_{i-1}, \tDum_{i}]} \con_i(\xSig(\tDumDum), \tDumDum) \},
\end{align*}
and we let
\begin{equation*}
    \valSRA[\tar_1,\con_1,\dots,\tar_N,\con_N; \tFin] = \val[\robSRA[\tar_1,\con_1,\dots,\tar_N,\con_N; \tFin]].
\end{equation*}

The following corollary allows us to compute the value function for a sequential reach-avoid task in terms of the value function for the same task when the first target is ignored, so long as the constraints do not grow stricter after the first target is reached.
If this assumption on the constraints does not hold, the monotonicity assumption \eqref{eqn:montonicity-condition} needed to apply Theorem \ref{thm:main-theorem} will not generally be satisfied.
A counter-example detailing what goes wrong in such a case will be detailed in Section \ref{sec:counter-example}.

\begin{corollary}[Sequential Reach-Avoid]\label{cor:ordered-reach-avoid}
Let $\tFin \in \R$, let $N \in \{2,3,\dots\}$, let $\tar_1,\dots,\tar_N: \Rn \tms \RleT \to \R$ all be continuous, and let $\con_1, \dots, \con_N: \Rn \tms \RleT \to \R$ all be continuous. 
If $\con_1 \le \con_2$, then
\begin{align*}
    \valSRA[\tar_1,\con_1,\dots,\tar_N,\con_N; \tFin] = \valRA[\tarSpec_1, \con_1; \tFin],
\end{align*}
where $\tarSpec_1 := \min\{\tar_1, \valSRA[\tar_2,\con_2,\dots,\tar_N,\con_N; \tFin] \}$.
\end{corollary}

\begin{proof}

    Let $\xSig \in \xSigs$ and $\tVal \in \RleT$.
    Let $\rob = \robSRA[\tar_1,\con_1,\dots,\tar_N,\con_N; \tFin]$ and $\rob_0 = \robSRA[\tar_2,\con_2,\dots,\tar_N,\con_N; \tFin]$.
    Then
    \begin{equation*}
        \rob(\xSig, \tVal) = \max_{\tDum \in [\tVal, \tFin]} \min\{\tar_1(\xSig(\tDum), \tDum), \min_{\tDumDum \in [\tVal, \tDum]} \con_1(\xSig(\tDumDum), \tDumDum),  \rob_0(\xSig, \tDum)\}.
    \end{equation*}
    The first part of Theorem \ref{thm:main-theorem} together with induction on $N$ then shows that $\rob$ is indeed continuous and past-independent.
    
    Moreover,
    \begin{align*}
        \min\{\min_{\tDumDum \in [\tVal, \tDum_1]}& \con_{1}(\xSig(\tDumDum), \tDumDum),  \rob_0(\xSig, \tDum_1)\}\\
        &\ge \min\{\min_{\tDumDum_1 \in [\tVal, \tDum_1]} \con_{1}(\xSig(\tDumDum_1), \tDumDum_1),\\
        &\qquad\,\quad \min_{\tDumDum_2 \in [\tDum_1, \tDum_2]} \con_{2}(\xSig(\tDumDum_2), \tDumDum_2), \rob_0(\xSig, \tDum_2)\}\\
        &\ge \min\{\min_{\tDumDum \in [\tVal, \tDum_2]} \con_{1}(\xSig(\tDumDum), \tDumDum), \rob_0(\xSig, \tDum_2)\}.
    \end{align*}
    The result then follows from Theorem \ref{thm:main-theorem}.
\end{proof}

The above Corollary leads to the following simple backward-inductive algorithm for computing the value function for a sequential reach-avoid task.

\begin{algorithm}
\caption{Sequential Reach-Avoid Algorithm}
\label{alg:example}
\begin{algorithmic}[1]
    \Require $N \in \N$, $\tFin \in \R$, continuous $\tar_1,\dots,\tar_N: \Rn \tms \RleT \to \R$, continuous $\con_1,\dots,\con_N: \Rn \tms \RleT \to \R$ satisfying $\con_1 \le \dots \le \con_N$
    \Ensure Compute $\valSRA[\tar_1, \con_1, \dots, \tar_N, \con_N; \tFin]$

    \State $V \gets \valRA[\tar_N, \con_N; \tFin]$

    \For{$i = N-1, N-2, \dots, 1$}
        \State $\tarSpec \gets \min\{\tar_i, V\}$
        \State $V \gets \valRA[\tarSpec, \con_i; \tFin]$
    \EndFor
    \State \Return $V$
\end{algorithmic}
\end{algorithm}

\subsection{Unordered Reach-Avoid}\label{sec:unordered-reach-avoid}

We now consider a different multiple-reach-avoid task in which the system must reach multiple targets in any order while obeying a constraint.
We refer to such a task an unordered reach-avoid task.

More precisely, given a terminal time $\tFin$, a nonempty, finite set $\tarFam$ of continuous target functions from $\Rn \tms \RleT$ into $\R$, and a continuous constraint function $\con: \Rn \tms \RleT \to \R$, we define the unordered reach-avoid performance functional $\robURA[\tarFam, \con; \tFin]: \xSigs \tms \RleT \to \R$ by $$\robURA[\tarFam, \con; \tFin](\xSig, \tVal) := \min_{\tar \in \tarFam} \robRA[\tar, \con; \tFin](\xSig, \tVal),$$
and we let
$$\valURA[\tarFam, \con; \tFin] = \val[\robURA[\tarFam, \con; \tFin]].$$

The following corollary gives us a method to compute the value function associated with an unordered reach-avoid task.

\begin{corollary}[Unordered Reach-Avoid]\label{cor:unordered-reach-avoid}
Let $\tFin \in \R$, let $\tarFam$ be a set of continuous functions from $\Rn \tms \RleT$ into $\R$, with $1 < |\tarFam| < \ii$, and let $\con: \Rn \tms \RleT \to \R$ be continuous.
Then
\begin{align*}
    \valURA[\tarFam, \con; \tFin] = \valRA[\tarSpec, \con; \tFin],
\end{align*}
where $\tarSpec: \Rn \tms \RleT \to \R$ is given by
$$\tarSpec(\xVal, \tVal) = \max_{\tar \in \tarFam}\min\{\tar(\xVal, \tVal), \valURA[\tarFam \setminus \{\tar\}, \con; \tFin](\xVal, \tVal)\}.$$
\end{corollary}

\begin{proof}
Let $\tar_1,\dots,\tar_N$ be an enumeration of the elements of $\tarFam$.
The result follows from Theorem \ref{thm:main-theorem-2} by using the lattice term function $\lpf = \min$ and the downstream performance functionals $\rob_{0,i} = \ii$ for each $i \in [N]$.
\end{proof}

\begin{algorithm}
\caption{Unordered Reach-Avoid Algorithm}
\label{alg:ura}
\begin{algorithmic}[1]
    \Require $N \in \N$, $\tFin \in \R$, continuous $\tar_1,\dots,\tar_N: \Rn \tms \RleT \to \R$, continuous $\con: \Rn \tms \RleT \to \R$
    \Ensure Compute $\valURA[\{\tar_1, \dots, \tar_N\}, \con; \tFin]$
    \For{$i = 1,\dots,N$}
        \State $V_{\{i\}} \gets \valRA[\tar_i, \con; \tFin]$
    \EndFor
    \For{$i = 2,\dots,N$}
        \For{$\mathcal{K} \in \binom{[N]}{i}$}
        \State $\tarSpec \gets \max_{k \in \mathcal{K}} \min\{\tar_k, V_{\mathcal{K} \setminus \{k\}} \}$
        \State $V_{\mathcal{K}} \gets \valRA[\tarSpec, \con; \tFin]$
        \EndFor
    \EndFor

    \State \Return $V_{[N]}$
\end{algorithmic}
\end{algorithm}

\subsection{Reach always-avoid}
We here note that we can also immediately derive the value decomposition for the reach always-avoid problem studied in \cite{hirsch-2206} using Theorem \ref{thm:main-theorem}.

Given a continuous target function $\tar: \Rn \tms \RleT \to \R$ and a continuous constraint function $\con: \Rn \tms \RleT \to \R$, we define the reach always-avoid performance functional
$\robRAA[\tar, \con; \tFin]: \xSigs \tms \RleT \to \R$ by
\begin{equation*}
    \robRAA[\tar,\con; \tFin](\xSig, \tVal) = \min\{\robR[\tar; \tFin](\xSig,\tVal), \robA[\con; \tFin](\xSig,\tVal) \},
\end{equation*}
and we let
\begin{equation*}
    \valRAA[\tar,\con; \tFin] = \val[\robRAA[\tar,\con; \tFin]].
\end{equation*}

\begin{corollary}[Reach always-avoid]
Let $\tFin \in \R$, and let $\tar,\con: \Rn \tms \RleT \to \R$ be continuous.
Then
\begin{align*}
    \valRAA[\tar, \con;\tFin] = \valRA[\tarSpec, \con; \tFin],
\end{align*}
where $\tarSpec = \min\{\tar, \valA[\con; \tFin]\}$.
\end{corollary}
\begin{proof}
Let $\rob = \robRAA[\tar, \con; \tFin]$ and $\rob_0 = \robA[\con;\tFin]$.
We check that the monotonicity condition \eqref{eqn:montonicity-condition} holds.
Indeed, let $\xSig \in \xSigs$, let $\tVal \le \tFin$, and let $\tDum_1, \tDum_2 \in \RleT$ be such that $\tVal \le \tDum_1 \le \tDum_2$.
Then
\begin{align*}
    \min\{&\min_{\tDumDum \in [\tVal, \tDum_1]} \conEval{\xSig}{\tDumDum}, \rob_0(\xSig, \tDum_1)\}\\
    &= \robA[\con;\tFin](\xSig, \tVal) \\
    &= \min\{\min_{\tDumDum \in [\tVal, \tDum_2]} \conEval{\xSig}{\tDumDum}, \rob_0(\xSig, \tDum_2)\}.
\end{align*}
Noting that
$$\rob(\xSig, \tVal) = \max_{\tDum \in [\tVal, \tFin]} \min\{\tar(\xSig(\tDum), \tDum), \min_{\tDumDum \in [\tVal, \tDum]} \con(\xSig(\tDumDum), \tDumDum), \rob_0(\xSig, \tDum)\},$$
the result then follows from Theorem \ref{thm:main-theorem}.
\end{proof}

\section{Importance of the monotonicity condition}\label{sec:counter-example}
In contrast to the one-player case, Theorem \ref{thm:main-theorem} generally may not hold in the two-player setting if the monotonicity condition \eqref{eqn:montonicity-condition} is not satisfied.
We here present an intuitive counter-example in which the monotonicity condition does not hold and the value decomposition fails.
The counter-example involves an ordered double reach-avoid game (as in Corollary \ref{cor:ordered-reach-avoid}), in which the obstacle in the downstream reach-avoid task is not a subset of the obstacle in the upstream reach-avoid task.

\begin{example}\label{ex:counter-example}
\newcommand{\conOneEval}[2]{\con_1 \lf( #1(#2), #2 \rg)}

In this counter-example, we show that conclusion of Theorem \ref{thm:main-theorem} need not hold if the monotonicity condition is not satisfied.
To make this point, we consider a sequential reach-avoid problem, as in Section \ref{sec:sequential-reach-avoid}.
This example is visualized in Figure \ref{fig:counter-example}.

Consider a scalar system $\dot{\xSig} = \dSig$, where the disturbance bound is $\dVals = [-1,1]$.
Note there is no explicit control action in the dynamics, but we could equivalently write the system as $\dot{\xSig} = \uSig + \dSig$, with $\uVals = \{0\}$ and $\dVals = [-1,1]$.

We define the goal regions $\goala = (-\ii, 1.5) \tms (1,2)$, $\goalb = (-\ii,2.75) \tms (3,4)$, $\goalone = \goala \cup \goalb$, and $\goaltwo = \R \tms (4,5)$, where the first index refers to the position in state-space and the second is time.
We define the obstacles $\obstacleone = \R \tms [-2,-1]$ and $\obstacletwo = (-\ii,1.5] \tms [2,3]$.

In this problem, the task is to be in $\goalone$ at some time $\tDum_1$, while avoiding $\obstacleone$ until time $\tDum_1$ (inclusively), and then to be in $\goaltwo$ at some time $\tDum_2 \ge \tDum_1$, while avoiding $\obstacletwo$ between times $\tDum_1$ and $\tDum_2$ (inclusively).

Recall that the sequential reach-avoid performance functional $\robSRA[\tar_1, \con_1, \tar_2, \con_2; \tFin]: \xSigs \tms \RleT \to \R$ is given by
\begin{align*}
    \robSRA[\tar_1, \con_1, \tar_2, \con_2; \tFin]&(\xSig, \tVal) = \\
    \max_{\tDum_1 \in [\tVal, \tFin]}\max_{\tDum_2 \in [\tDum_1, \tFin]} \min\{& \tar_1(\xSig(\tDum_1), \tDum_1), \min_{\tDumDum_1 \in [\tVal, \tDum_1]} \con_1(\xSig(\tDumDum_1), \tDumDum_1),\\
    &\tar_2(\xSig(\tDum_2), \tDum_2), \min_{\tDumDum_2 \in [\tDum_1, \tDum_2]} \con_2(\xSig(\tDumDum_2), \tDumDum_2)  \},
\end{align*}
where $\tFin = 5$ and $\tar_1$, $\tar_2$, $\con_1$, $\con_2$ are the signed distance functions to $\goalone$, $\goaltwo$, $\obstacleone^C$, and $\obstacletwo^C$, respectively.

Note that the disturbance can ensure via the signal $\dSig_{-1}(\cdot) := -1$ that the system will enter the interior of the obstacle $\obstacletwo$ from anywhere in $\bargoala$, without first reaching $\bargoaltwo$.
Thus, we have $\valRA[\tar_2, \con_2;\tFin] < 0$ on $\bargoala$.

Now consider the initial state $(\xNot, \tNot) := (0,0)$ and the disturbance signal $\dSig_1(\cdot) := 1$.
At each time $\tVal \in [0, \tFin]$, one can check that we either have $\valRA[\tar_2, \con_2;\tFin](\traj{0,0}{\dSig_1}(\tVal), \tVal) < 0$ or $\tar_1(\traj{0,0}{\dSig_1}(\tVal), \tVal) < 0$.
Thus, letting $\tarSpec := \min\{\tar_1, \valRA[\tar_2, \con_2;\tFin]\}$, we have $\tarSpec(\traj{0,0}{\dSig_1}(\tVal), \tVal) < 0$ for all $\tVal \in [0, \tFin]$.
It thus follows that 
\begin{equation*}
    \valRA[\tarSpec, \con_1; \tFin](0,0) < 0.
\end{equation*}

However, from this initial state and time, the ordered double reach-avoid task will be completed, regardless of the disturbance signal.
More precisely, we have
$$\valSRA[\tar_1, \con_1, \tar_2, \con_2; \tFin](0,0) \ge 0.$$
(This above inequality is in fact strict, though we do not need this fact for the contradiction.)

The apparent contradiction with Theorem \ref{thm:main-theorem} is resolved by noticing that Theorem \ref{thm:main-theorem}'s hypotheses do not hold for this problem.
In particular, the monotonicity condition \eqref{eqn:montonicity-condition} in Theorem \ref{thm:main-theorem} is not satisfied.
To see this, note that for the trajectory $\xSig_0 := \traj{0,0}{0}$ (i.e. the trajectory starting from the origin under the disturbance signal $\dSig_0(\cdot) := 0)$, we have
\begin{align*}
    &\min\{\min_{\tDumDum \in [0, 0]} \conOneEval{\xSig_0}{\tDumDum}, \robRA[\tar_2, \con_2; \tFin](\xSig_0, 0)\}\\
    &=\min\{\conOneEval{\xSig_0}{0}, \robRA[\tar_2, \con_2; \tFin](\xSig_0, 0)\}\\
    &< 0,
\end{align*}
because the trajectory $\xSig_0$ intersects $\obstacletwo$ after time $0$, but
\begin{align*}
    &\min\{\min_{\tDumDum \in [0, 4]} \conOneEval{\xSig_0}{\tDumDum}, \robRA[\tar_2, \con_2; \tFin](\xSig_0, 4)\} > 0,
\end{align*}
because the trajectory $\xSig_0$ is within $\goaltwo$ at time $4.5$ without intersecting $\obstacleone$ on the time interval $[0,4]$ or intersecting $\obstacletwo$ after time $4$. \blackbox
\end{example}

\begin{remark}
    The above counter-example can also be seen as an example justifying the need for the hypothesis $\con_1 \le \con_2$ in Corollary \ref{cor:ordered-reach-avoid}.
    Indeed, the task studied in the counter-example is an ordered double reach-avoid task, but here $\con_1 \not\le \con_2$ because $\obstacletwo \not\sbs \obstacleone$.
\end{remark}
\begin{figure}
    \centering
    \includegraphics[width=1.0\linewidth]{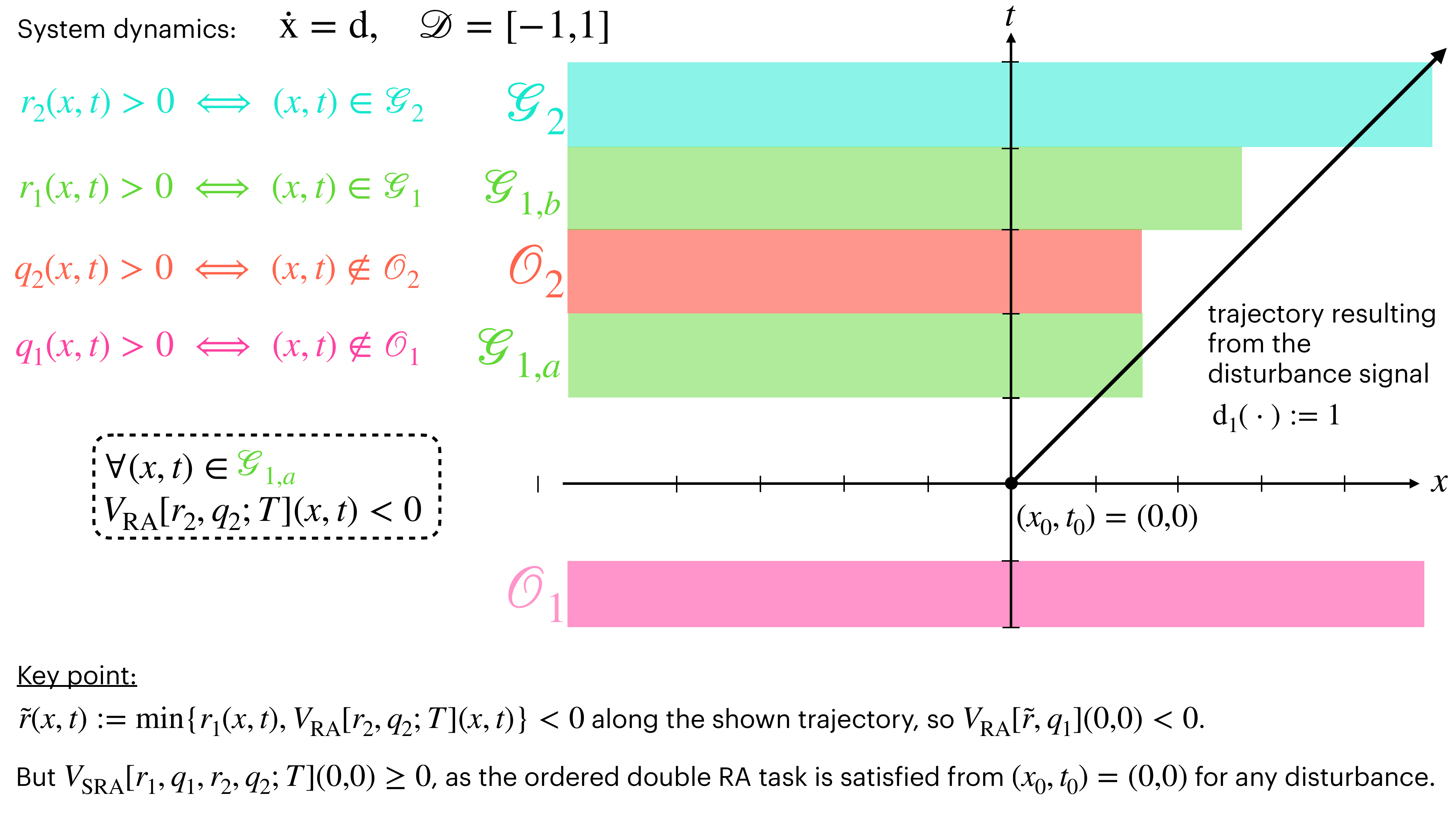}
    \caption{Visualization of Counter-Example \ref{ex:counter-example}.
    We show that without the monotonicity condition \eqref{eqn:montonicity-condition}, the result from Theorem \ref{thm:main-theorem} does not in general hold.
    In Counter-Example 
    \ref{ex:counter-example}, we demonstrate this point via an ordered double reach-avoid problem.
    Specifically, the task is to be in $\goalone := \goala \cup \goalb$ (green) at some time $\tDum_1$, while avoiding $\obstacleone$ (pink) until time $\tDum_1$, and then to be in $\goaltwo$ (cyan) at some time $\tDum_2 \ge \tDum_1$, while avoiding $\obstacletwo$ between times $\tDum_1$ and $\tDum_2$ (orange), with the final time for completion of the task being $\tFin = 5$.
    Because the disturbance can ensure that the system enters $\obstacletwo^\circ$ from any point in $\bargoala$, without first reaching $\bargoaltwo$, we have $\valRA[\tar_2,\con_2; \tFin](\xVal, \tVal) < 0$ for every $(\xVal, \tVal) \in \bargoala$.
    As such, we have $\tarSpec := \min\{\tar_1, \valRA[\tar_2,\con_2; \tFin] \} < 0$ along the trajectory $\traj{0,0}{\dSig_1}$, where $\dSig_1(\cdot) := 1$.
    This implies that $\valRA[\tarSpec, \con_1; \tFin](0,0) < 0$.
    But regardless of the disturbance signal, the ordered double reach-avoid task will be satisfied from the initial state $\xNot = 0$ and time $\tNot = 0$, i.e. $\valSRA[\tar_1, \con_1, \tar_2, \con_2; \tFin](0,0) \ge 0$.
    }
    \label{fig:counter-example}
\end{figure}

\appendix
\subsection{Proof of Theorem \ref{thm:main-theorem}}\label{sec:appendix-main-theorem}

\begin{proof}
\newcommand{\tJ}{\rob_0}

\newcommand{\firstStrat}{\alpha}
\newcommand{\secondStrat}[1]{\beta[#1]}
\newcommand{\composedStrat}{\dStrat}

\newcommand{\vSig}{\mathrm{v}}
\newcommand{\aSig}{\mathrm{a}}
\newcommand{\bSig}{\mathrm{b}}

\newcommand{\switchTime}[1]{\tDum_{#1}}
\newcommand{\switchState}[1]{\xDum_{#1}}
\newcommand{\switchSet}[1]{\mathcal{S}_{#1}}

\newcommand{\firstTraj}[1]{\traj{x,t}{#1,\firstStrat}}
\newcommand{\secondTraj}[1]{\traj{\switchState{\vSig},\switchTime{\vSig}}{#1,\secondStrat{\switchState{\vSig},\switchTime{\vSig}}}}
\newcommand{\composedTraj}[1]{\traj{x,t}{#1,\composedStrat}}

\newcommand{\ep}{\epsilon}

% trajectories
\newcommand{\xtua}{\traj{\xVal, \tVal}{\uSig, \firstStrat}}
\newcommand{\xtva}{\traj{\xVal, \tVal}{\vSig, \firstStrat}}
\newcommand{\xtud}{\traj{\xVal, \tVal}{\uSig, \composedStrat}}
\newcommand{\xtvd}{\traj{\xVal, \tVal}{\vSig, \composedStrat}}
\newcommand{\ysub}{\traj{\xDum, \tDum}{\uSig, \secondStrat{\xDum, \tDum}}}
\newcommand{\ysvb}{\traj{\switchTime{\vSig}, \switchState{\vSig}}{\vSig, \secondStrat{\switchTime{\vSig}, \switchState{\vSig}}}}

\newcommand{\xtad}{\traj{\xVal, \tVal}{\aSig, \dStrat}}
\newcommand{\ysbda}{\traj{\xDum, \tDum}{\bSig, \dStrat_\aSig}}
\newcommand{\ysuda}{\traj{\xDum, \tDum}{\uSig, \dStrat_\aSig}}

First, note that it suffices to assume that $\rob_0$, $\tar$, and $\con$ have bounded images, for otherwise we could replace these functions with $\homeo  \circ \rob_0$, $\homeo \circ \tar$, and $\homeo \circ \con$, where $\homeo : \extended \to [-1,1]$ is defined by
\begin{equation*}
    \homeo(y) = \begin{cases}
        -1 & y = -\ii,\\
        \tanh{y} & -\ii < y < \ii,\\
        1 & y = \ii.
    \end{cases}
\end{equation*}
Having proven the corresponding result under this substitution, the conclusion then follows from the fact that $\homeo$ is homeomorphic and increasing (so that it is invertible and commutes with the relevant $\min$, $\max$, $\sup$, and $\inf$ operations).

Also, in the remainder of this proof, given an $\xVal \in \Rn$, $\tVal \in \R$, $\uSig \in \uSigs$, and $\dStrat \in \dStrats$, we will define $\traj{\xVal, \tVal}{\uSig, \dStrat} = \stratTraj$.

[Part 1]
First, note that past-independence of $\rob$ follows directly from $\eqref{eqn:main-theorem-composite-performance-funcitonal}$.
We now prove continuity of $\rob$.
Fix some $\tInit < \tFin$.
Let $\tVal \in (\tInit, \tFin]$ and $\xSig_1, \xSig_2 \in \xSigs$.
We have by the standard inequality for the difference of maximums (equivalently minimums) that
\begin{align*}
    |\rob(\xSig_1, \tVal) - \rob(\xSig_2, \tVal)| \le&\\
    \max_{\tDum \in [\tInit, \tFin]} \max\{
    &|\tarEval{\xSig_1}{\tDum} - \tarEval{\xSig_2}{\tDum}|,\\
    &|\conEval{\xSig_1}{\tDum} - \conEval{\xSig_2}{\tDum}|,\\
    &|\rob_0(\xSig_1,\tDum) - \rob_0(\xSig_2,\tDum)|\}.
\end{align*}
Because $\tar$, $\con$, and $\rob_0$ are each continuous, it follows that the right-hand-side is continuous in $\xSig_1$ and $\xSig_2$.
Since this bound holds uniformly for each $\tVal \in (\tInit, \tFin]$, it follows that $\rob(\xSig, \tVal)$ is continuous in $\xSig$, locally uniformly in $\tVal$.

Moreover, for each $\xSig \in \xSigs$ and $\tVal, \tDum \in \RleT$ for which $\tVal \le \tDum$ we have
\begin{align*}
    |\rob(\xSig, \tVal) - \rob(\xSig, \tDum)| \le&\\
    \max_{\tDumDum_1, \tDumDum_2 \in [\tVal, \tDum]} \max\{
    &|\tarEval{\xSig}{\tDumDum_1} - \tarEval{\xSig}{\tDumDum_2}|,\\
    &|\conEval{\xSig}{\tDumDum_1} - \conEval{\xSig}{\tDumDum_2}|,\\
    &|\rob_0(\xSig,\tDumDum_1) - \rob_0(\xSig,\tDumDum_2)|\}.
\end{align*}
Thus $\rob(\xSig, \cdot)$ is continuous for each $\xSig \in \xSigs$ due to the continuity of $\tar$, $\con$, and $\rob_0$.
Continuity of $\rob$ follows.

[Part 2]
We now prove \eqref{eqn:main-theorem-conclusion}.

($\le$) This direction proceeds by building a near-optimal composite strategy $\composedStrat$ from a primary strategy $\firstStrat$ and a family of secondary strategies $\secondStrat{\xDum,\tDum}$.
    
Fix $\ep > 0$, $\xVal \in \Rn$, and $\tVal \le \tFin$.
Select an adversary strategy $\firstStrat \in \dStrats$ such that
\begin{equation}
    \supu \robRA[\tarSpec,\con; \tFin](\xtua, \tVal) \le \valRA[\tarSpec, \con; \tFin](\xVal, \tVal) + \ep.
\end{equation}
For each $\xDum \in \Rn$ and $\tDum \le \tFin$, choose $\secondStrat{\xDum,\tDum} \in \dStrats$ such that
\begin{equation}
    \supu \tJ\lf(\ysub, \tDum \rg) \le \val[\tJ](\xDum, \tDum) + \ep.
\end{equation}

For each $\uSig \in \uSigs$, we define a set of candidate switch times
\begin{equation}
    \switchSet{\uSig} := \lf\{ \tDum \in [\tVal,\tFin] :  \tarEval{\firstTraj{\uSig}}{\tDum} \ge \valRA[\tarSpec, \con; \tFin](\xVal, \tVal) + 2\ep \rg\},
\end{equation}
we set the actual switch time as 
\begin{equation*}
\switchTime{\uSig} =
    \begin{cases}
       \min \switchSet{\uSig} & \switchSet{\uSig} \ne \varnothing\\
       \ii & \switchSet{\uSig} = \varnothing,
    \end{cases}
\end{equation*}
and we set the switch state as $\switchState{\uSig} := \firstTraj{\uSig}(\switchTime{\uSig})$ when $\switchTime{\uSig} < \ii$. (Note that $\switchSet{\uSig}$ is indeed closed due to continuity of $\tar$ and $\firstTraj{\uSig}$, so the minimum of this set exists whenever the set is non-empty).

We define the composite strategy $\composedStrat \in \dStrats$ by for each $\uSig \in \uSigs$ letting $\composedStrat(\uSig) \in \dSigs$ be given by
    \begin{equation*}
        \composedStrat(\uSig)(\tVal) = \begin{cases}
            \firstStrat(\uSig)(\tVal)
            & \tVal < \switchTime{\uSig}
            \\
            \secondStrat{\switchState{\uSig}, \switchTime{\uSig}}(\uSig)(\tVal) & \tVal \ge \switchTime{\uSig}
        \end{cases}
    \end{equation*}
    when $\switchTime{\uSig} < \ii$ and setting $\composedStrat(\uSig) = \firstStrat(\uSig)$ otherwise.

We show that $\composedStrat$ is indeed non-anticipative.
    Let $\tDum \in \R$, and suppose $\uSig_1(\tDumDum) = \uSig_2(\tDumDum)$ for a.e. on $\tDumDum \le \tDum$.
    First, assume that $\min\{\switchTime{\uSig_1},\switchTime{\uSig_2}\} \le \tDum$.
    We show $\switchTime{\uSig_1} = \switchTime{\uSig_2}$ and $\switchState{\uSig_1} = \switchState{\uSig_2}$.
    Assume $\switchTime{\uSig_1} \le \switchTime{\uSig_2}$.
    Then $\switchTime{\uSig_1} \le \tDum$, so that 
    \begin{equation}
        \firstTraj{\uSig_1}(\switchTime{\uSig_1}) = \firstTraj{\uSig_2}(\switchTime{\uSig_1})
    \end{equation} by non-anticipativity of $\firstStrat$.
    Thus, because  $\switchTime{\uSig_1} \in \switchSet{\uSig_1}$ (as $\switchTime{\uSig_1} \le \tDum < \ii$), we also have $\switchTime{\uSig_1} \in \switchSet{\uSig_2}$.
    But since $\switchTime{\uSig_2}$ is the minimal element of $\switchSet{\uSig_2}$ and $\switchTime{\uSig_1} \le \switchTime{\uSig_2}$, we indeed have  $\switchTime{\uSig_1} = \switchTime{\uSig_2}$, and also
    $\switchState{\uSig_1} = \switchState{\uSig_2}$.
    The same result follows if instead $\switchTime{\uSig_2} \le \switchTime{\uSig_1}$.

    Thus, either (i) $\switchTime{\uSig_1} > \tDum$ and $\switchTime{\uSig_2} > \tDum$ or (ii) $\switchTime{\uSig_1} = \switchTime{\uSig_2} \le \tDum$ and $\switchState{\uSig_1} = \switchState{\uSig_2}$.
    In case (i), non-anticipativity of $\firstStrat$ implies that $$\composedStrat(\uSig_1)(\tDumDum) = \firstStrat(\uSig_1)(\tDumDum) = \firstStrat(\uSig_2)(\tDumDum) = \composedStrat(\uSig_2)(\tDumDum)$$
    for a.e. $\tDumDum \le \tDum$. 
    In case (ii), non-anticipativity of $\firstStrat$ implies that $$\composedStrat(\uSig_1)(\tDumDum) = \firstStrat(\uSig_1)(\tDumDum) = \firstStrat(\uSig_2)(\tDumDum) = \composedStrat(\uSig_2)(\tDumDum)$$ 
    for a.e.$ \tDumDum \le \switchTime{\uSig_1}$ and non-anticipativity of $\secondStrat{\switchState{\uSig_1}, \switchTime{\uSig_1}} = \secondStrat{\switchState{\uSig_2}, \switchTime{\uSig_2}}$ implies that 
    $$\composedStrat(\uSig_1)(\tDumDum) = \secondStrat{\switchState{\uSig_1}, \switchTime{\uSig_1}}(\uSig_1)(\tDumDum) = \secondStrat{\switchState{\uSig_2}, \switchTime{\uSig_2}}(\uSig_2)(\tDumDum) = \composedStrat(\uSig_2)(\tDumDum)$$
    for a.e. $\tDumDum \in [\switchTime{\uSig_1},\tDum]$.
    It follows that $\composedStrat \in \dStrats$.

    Now, fix a control signal $\vSig \in \uSigs$.
    There are two cases: (i) $\switchTime{\vSig} < \ii$ and (ii) $\switchTime{\vSig} = \ii$.
    
    (i) First suppose that $\switchTime{\vSig} < \ii$.
    Note that for each $\tDumDum \in \R$ we have
    \begin{equation*}
        \composedStrat(\vSig)(\tDumDum) = \begin{cases}
            \firstStrat(\vSig)(\tDumDum) & \tDumDum < \switchTime{\vSig}, \\
            \secondStrat{\switchState{\vSig},\switchTime{\vSig}}(\vSig)(\tDumDum) & \tDumDum \ge \switchTime{\vSig}.
        \end{cases}
    \end{equation*}
    Thus for all $\tDumDum \le \switchTime{\vSig}$,
    \begin{equation}
        \composedTraj{\vSig}(\tDumDum) = \firstTraj{\vSig}(\tDumDum),
    \end{equation}
    and for all $\tDumDum \ge \switchTime{\vSig}$,
    \begin{equation}
        \composedTraj{\vSig}(\tDumDum) = \secondTraj{\vSig}(\tDumDum).
    \end{equation}
    
    We observe that since $\switchTime{\vSig} \in \switchSet{\vSig}$
    \begin{align*}
        &\tarEval{\firstTraj{\vSig}}{\switchTime{\vSig}} \ge \valRA[\tarSpec, \con; \tFin](\xVal, \tVal) + 2\ep\\
        &\ge \min\{\tarEval{\firstTraj{\vSig}}{\switchTime{\vSig}}, \min_{\tDumDum \in [\tVal, \switchTime{\vSig}]} \conEval{\xtva}{\tDumDum},  \\
        &\qquad\qquad\qquad\qquad\qquad\quad \val[\rob_0](\switchState{\vSig}, \switchTime{\vSig}) \} + \ep,
    \end{align*}
    so that
    \begin{align*}
        &\tarEval{\firstTraj{\vSig}}{\switchTime{\vSig}} \\
        &\ge
        \min\{\min_{\tDumDum \in [\tVal, \switchTime{\vSig}]} \conEval{\xtva}{\tDumDum},  \val[\rob_0](\switchState{\vSig}, \switchTime{\vSig}) \} + \ep.
    \end{align*}

    Let $\tDum \in [\tVal, \tFin]$.
    If $\tDum \in [\tVal, \switchTime{\vSig})$, then
    \begin{align*}
        &\min\{\tarEval{\xtvd}{\tDum}, \min_{\tDumDum \in [\tVal, \tDum]} \conEval{\xtvd}{\tDumDum},  \rob_0(\xtvd, \tDum)\} \\
        &\le \tarEval{\xtva}{\tDum}\\
        &< \valRA[\tarSpec, \con; \tFin](\xVal, \tVal) + 2\ep.
    \end{align*}
    If instead we have $\tDum \in [\switchTime{\vSig}, \tFin]$, then
    \begin{align*}
        &\min\{\tarEval{\xtvd}{\tDum}, \min_{\tDumDum \in [\tVal, \tDum]} \conEval{\xtvd}{\tDumDum},  \rob_0(\xtvd, \tDum)\} \\
        &\le \min\{\min_{\tDumDum \in [\tVal, \switchTime{\vSig}]} \conEval{\xtvd}{\tDumDum},  \rob_0(\xtvd, \switchTime{\vSig})\} \\
        &= \min\{\min_{\tDumDum \in [\tVal, \switchTime{\vSig}]} \conEval{\xtva}{\tDumDum},  \rob_0(\secondTraj{\vSig}, \switchTime{\vSig})\} \\
        &\le \min\{\min_{\tDumDum \in [\tVal, \switchTime{\vSig}]} \conEval{\xtva}{\tDumDum},  \val[\rob_0](\switchState{\vSig}, \switchTime{\vSig})\} + \ep,\\
        &\le \min\{\tarEval{\xtva}{\switchTime{\vSig}}, \min_{\tDumDum \in [\tVal, \switchTime{\vSig}]} \conEval{\xtva}{\tDumDum},  \\
        &\qquad\qquad\qquad\qquad\qquad\quad\val[\rob_0](\xtva(\switchTime{\vSig}), \switchTime{\vSig})\} + 2\ep\\
        &\le \valRA[\tarSpec, \con; \tFin](\xVal, \tVal) + 3\ep,
    \end{align*}
    where the first inequality follows from the monotonicity assumption \eqref{eqn:montonicity-condition} and the following equality follows from past-independence of $\rob_0$.
    
    Thus
    \begin{align*}
        \rob(\xtvd, \tVal) &= \max_{\tDum \in [\tVal, \tFin]} \min\{\tarEval{\xtvd}{\tDum}, \\
        &\qquad\qquad\qquad\min_{\tDumDum \in [\tVal, \tDum]} \conEval{\xtvd}{\tDumDum},  \rob_0(\xtvd, \tDum)\} \\
        &\le \valRA[\tarSpec, \con; \tFin](\xVal, \tVal) + 3\ep.
    \end{align*}

    (ii) Now instead suppose we have $\switchTime{\vSig} = \ii$.
    Then
    \begin{align*}
        \rob(\xtvd, \tVal) &= \max_{\tDum \in [\tVal, \tFin]} \min\{\tarEval{\xtvd}{\tDum},\\
        &\qquad\qquad\qquad\min_{\tDumDum \in [\tVal, \tDum]} \conEval{\xtvd}{\tDumDum},  \rob_0(\xtvd, \tDum)\} \\
        &\le \valRA[\tarSpec, \con; \tFin](\xVal, \tVal) + 2\ep.
    \end{align*}

    In either case, we have $\rob(\xtvd, \tVal) \le \valRA[\tarSpec, \con; \tFin](\xVal, \tVal) + 3\ep$.
    Because $\vSig \in \uSigs$ was arbitrary, we can conclude $\val[\rob](\xVal, \tVal) \le \valRA[\tarSpec, \con; \tFin](\xVal, \tVal) + 3\ep$.

    ($\ge$) In this direction, we build a near-optimal control signal $\vSig$ from a primary signal $\aSig$ and a secondary signal $\bSig$ for responding to a near-optimal adversary strategy $\dStrat$.

    Fix $\ep > 0$, $\xVal \in \Rn$, and $\tVal \le \tFin$.
    First, choose $\dStrat \in \dStrats$.
    Next, choose $\aSig \in \uSigs$ and $\tDum \in [\tVal,\tFin]$ such that
    \begin{align*}
        &\min\lf\{ \tarSpecEval{\xtad}{\tDum}, \min_{\tDumDum \in [\tVal,\tDum]} \conEval{\xtad}{\tDumDum} \rg\} \\
        &\ge \valRA[\tarSpec, \con; \tFin](\xVal, \tVal) - \ep.
    \end{align*}

    For each $\uSig \in \uSigs$, let $\uSig_\aSig \in \uSigs$ be given by
    \begin{equation*}
        \uSig_\aSig(\tDumDum) = \begin{cases}
            \aSig(\tDumDum) & \tDumDum < \tDum\\
            \uSig(\tDumDum) & \tDumDum \ge \tDum.
        \end{cases}
    \end{equation*}
    Define the adversary strategy $\dStrat_\aSig \in \dStrats$ by
    \begin{equation*}
        \dStrat_\aSig(\uSig) = \dStrat(\uSig_\aSig),
    \end{equation*}
    where $\dStrat_\aSig$ is non-anticipative because $\dStrat$ is.
    Finally, let $\xDum = \xtad(\tDum)$, and select $\bSig \in \uSigs$ such that
    \begin{equation*}
        \rob_0(\ysbda,\tDum) \ge \supu \rob_0(\ysuda,\tDum) - \ep.
    \end{equation*}

    Define the control signal $\vSig \in \uSigs$ by
    \begin{equation*}
        \vSig(\tDumDum) = \begin{cases}
            \aSig(\tDumDum) & \tDumDum < \tDum, \\
            \bSig(\tDumDum) & \tDumDum \ge \tDum.
        \end{cases}
    \end{equation*}
    
    By non-anticipativity of $\dStrat$, we have $\xtad(\tDumDum) = \xtvd(\tDumDum)$ for all $\tDumDum \le \tDum$.
    In particular, $\xDum = \xtvd(\tDum)$.
    Also note that for all $\tDumDum \ge \tDum$, we have
    $\ysbda(\tDumDum) = \xtvd(\tDumDum)$ by definition of $\dStrat_\aSig$.
    It follows that
    \begin{align*}
        &\supu \rob(\xtud, \tVal)\\
        &\ge \min\lf\{ \tar(\xDum,\tDum), \min_{\tDumDum \in [\tVal,\tDum]} \conEval{\xtvd}{\tDumDum}, \rob_0(\xtvd,\tDum)\rg\}\\
        &= \min\lf\{ \tar(\xDum,\tDum), \min_{\tDumDum \in [\tVal,\tDum]} \conEval{\xtad}{\tDumDum}, \rob_0(\ysbda,\tDum)\rg\}\\
        &\ge \min\lf\{ \tar(\xDum,\tDum), \min_{\tDumDum \in [\tVal,\tDum]} \conEval{\xtad}{\tDumDum}, \val[\rob_0](\xDum, \tDum) \rg\} - \ep\\
        &= \min\lf\{ \tarSpec(\xDum,\tDum), \min_{\tDumDum \in [\tVal,\tDum]} \conEval{\xtad}{\tDumDum}\rg\} - \ep\\
        &\ge \valRA[\tarSpec, \con; \tFin](\xVal, \tVal) - 2\ep,
    \end{align*}
    where the first equality follows from past independence of $\rob_0$.
    Because $\dStrat \in \dStrats$ was arbitrary, we can conclude that
    $\val[\rob](\xVal, \tVal) \ge \valRA[\tarSpec, \con; \tFin](\xVal, \tVal) - 2\ep$.
\end{proof}

\subsection{Proof of Theorem \ref{thm:main-theorem-2}}\label{sec:appendix-main-theorem-2}

We will use the following basic facts about lattice term functions in the upcoming proof.
\begin{lemma}\label{lemma:lattice-term-functions}
    Let $\lpf:\extended^N \to \extended$ be a lattice term function.
    \begin{itemize}
        \item[\textit{(i)}]
        Let $y \in \extended^N$. There is some $i \in [N]$ such that $y_i = \lpf(y)$.
        \item[\textit{(ii)}]
        Let $y,z \in \extended^N$, and let $c \in \extended$.
        If $\lpf(y) \ge c$ and $z_i \ge c$ for each $i \in [N]$ for which $y_i \ge c$, then $\lpf(z) \ge c$.
        \item[\textit{(iii)}] Let $y \in \extended^N$, and let $c \in \extended$.
        Then $\lpf(\min\{c,y_1\},\dots,\min\{c,y_N\}) = \min\{c,\lpf(y)\}$.
    \end{itemize}
\end{lemma}
\begin{proof} 
The statements all clearly hold for $0$-ary lattice term functions.
Suppose they also hold for $k$-ary lattice term functions, and suppose $\lpf$ is a $(k+1)$-ary lattice term function.
Then $\lpf = \min\{\lpf_1, \lpf_2\}$ or $\lpf = \max\{\lpf_1, \lpf_2\}$ for some $k$-ary lattice term functions $\lpf_1$ and $\lpf_2$.

To show (i), choose $i,j \in [N]$ such that $y_i = \lpf_1(y)$ and $y_j = \lpf_2(y)$.
Then $\lpf(y) = \lpf_1(y) = y_i$ or $\lpf(y) = \lpf_2(y) = y_j$.

To show (ii),
first suppose that $\lpf = \min\{\lpf_1, \lpf_2\}$.
Since $\lpf(y) \ge c$, then $\lpf_1(y) \ge c$ and $\lpf_2(y) \ge c$, so that $\lpf_1(z) \ge c$ and $\lpf_2(z) \ge c$.
But this implies that $\lpf(z) \ge c$ as well.

Instead suppose that $\lpf = \max\{\lpf_1, \lpf_2\}$.
Then $\lpf_1(y) \ge c$ or $\lpf_2(y) \ge c$, so that $\lpf_1(z) \ge c$ or $\lpf_2(z) \ge c$.
But this again implies that $\lpf(z) \ge c$ as well. 

To show (iii), for convenience, let 
$$\min\{c,y\} = (\min\{c,y_1\},\dots,\min\{c,y_N\}).$$
Again suppose that $\lpf = \min\{\lpf_1, \lpf_2\}$.
Then 
\begin{align*}
    \lpf(\min\{c,y\}) &= \min\{\lpf_1(\min\{c,y\}), \lpf_2(\min\{c,y\})\} \\
    &= \min\{\min\{c, p_1(y)\}, \min\{c, p_2(y)\}\} \\
    &= \min\{c, p(y)\}.
\end{align*}
On the other hand, if $\lpf = \max\{\lpf_1, \lpf_2\}$, then 
\begin{align*}
    \lpf(\min\{c,y\}) &= \max\{\lpf_1(\min\{c,y\}), \lpf_2(\min\{c,y\})\} \\
    &= \max\{\min\{c, p_1(y)\}, \min\{c, p_2(y)\}\} \\
    &= \min\{c, \max\{p_1(y), p_2(y)\}\} = \min\{c, p(y)\}.
\end{align*}

\qedhere
\end{proof}

\begin{proof}[Proof of Theorem \ref{thm:main-theorem-2}]
\newcommand{\tarOneEval}[2]{\tar_1 \lf( #1(#2), #2 \rg)}
\newcommand{\tarTwoEval}[2]{\tar_2 \lf( #1(#2), #2 \rg)}
\newcommand{\tarDum}{\tar'}

Because $\rob$ is the composition of a finite family of continuous and past-independent performance functionals with the continuous function $\lpf$, it is indeed the case that $\rob$ is continuous and past-independent.

For convenience, we will establish the following notation.
First, we let $\tar^* = \max_{i \in [N]}\tar_i$.
Next, for each $i \in [N]$, we let $\rob_{\lnot i} = (\rob_1,\dots,\rob_{i-1}, \rob_{i+1}, \dots, \rob_N)$, and we let $\lpf_i: \extended \tms \extended^{N-1} \to \extended$ be defined by $\lpf_i(y_i, (y_1,\dots,y_{i-1},y_{i+1}, \dots y_N)) = \lpf(y)$.
In other words, $\lpf_i$ rearranges the arguments of $\lpf$ so that the $i$-th entry appears first and the remaining entries are grouped into a vector of size $N-1$.
Finally, we define the operator notation $\lat_{i \in [N]}$ by setting
\begin{equation*}
    \lat_{i \in [N]} y_i = \lpf(y_1,\dots,y_N)
\end{equation*}
for each $y \in \extended^\N$.
Similarly, given any performance functionals $\hat{\rob}_1,\dots,\hat{\rob}_N: \xSigs \tms \RleT \to \extended$, we define $\lat_{i \in [N]} \hat{\rob}_i: \xSigs \tms \RleT \to \extended$ as the point-wise application of the operator, i.e.
\begin{equation*}
    \Big(\lat_{i \in [N]} \hat{\rob}_i\Big)(\xSig, \tVal) = \lat_{i \in [N]} \big(\hat{\rob}_i(\xSig, \tVal)\big).
\end{equation*}

\noindent[Step 1] We show that the monotonicity condition \eqref{eqn:montonicity-condition} holds with $\rob_0 := \rob$.
Indeed, let $\xSig \in \xSigs$, let $\tVal \le \tFin$, and let $\tDum_1, \tDum_2 \in \RleT$ be such that $\tVal \le \tDum_1 \le \tDum_2$.
Then for each $i \in [N]$ we have
\begin{align*}
    \min\{&\min_{\tDumDum \in [\tVal, \tDum_1]} \conEval{\xSig}{\tDumDum}, \rob_{i}(\xSig, \tDum_1)\}\\
    &= \max_{\tDum \in [\tDum_1, \tFin]} \min\{\min_{\tDumDum \in [\tVal, \tDum]} \conEval{\xSig}{\tDumDum}, \tar_i(\xSig(\tDum),\tDum), \rob_{0,i}(\xSig, \tDum)\} \\
    &\ge \max_{\tDum \in [\tDum_2, \tFin]} \min\{\min_{\tDumDum \in [\tVal, \tDum]} \conEval{\xSig}{\tDumDum}, \tar_i(\xSig(\tDum),\tDum), \rob_{0,i}(\xSig, \tDum)\} \\
    &= \min\{\min_{\tDumDum \in [\tVal, \tDum_2]} \conEval{\xSig}{\tDumDum}, \rob_{i}(\xSig, \tDum_2)\}.
\end{align*}
Thus by Lemma \ref{lemma:lattice-term-functions}(iii), we have
\begin{align*}
    \min\{&\min_{\tDumDum_1 \in [\tVal, \tDum_1]} \conEval{\xSig}{\tDumDum_1}, \rob(\xSig, \tDum_1)\}\\
    &= \lat_{i \in [N]} \min\{\min_{\tDumDum \in [\tVal, \tDum_1]} \conEval{\xSig}{\tDumDum}, \rob_{i}(\xSig, \tDum_1)\}\\
    &\ge \lat_{i \in [N]} \min\{\min_{\tDumDum \in [\tVal, \tDum_2]} \conEval{\xSig}{\tDumDum}, \rob_{i}(\xSig, \tDum_2)\}\\
    &= \min\{\min_{\tDumDum_2 \in [\tVal, \tDum_2]} \conEval{\xSig}{\tDumDum_2}, \rob(\xSig, \tDum_2)\}.
\end{align*}

\noindent[Step 2] Fix $\xSig \in \xSigs$ and $\tVal \le \tFin$.
We claim that 
$$\rob(\xSig, \tVal) = \max_{\tDum \in [\tVal, \tFin]} \min\{\tar^*(\xSig(\tDum), \tDum), \min_{\tDumDum \in [\tVal, \tDum]} \con(\xSig(\tDumDum), \tDumDum), \rob(\xSig, \tDum)\}.$$

\noindent($\le$)
Let $\mathcal{I} = \{i \in [N] \mid \rob_i(\xSig, \tVal) \ge \rob(\xSig, \tVal)\}$ and
$$S = \{\tDum \in [\tVal, \tFin] \mid \min\{\tar^*(\xSig(\tDum), \tDum), \min_{\tDumDum \in [\tVal, \tDum
]} \con(\xSig(\tDumDum), \tDumDum) \} \ge \rob(\xSig, \tVal) \}.$$
First, note that $S$ is non-empty, for if it were empty then we would have
\begin{align*}
    &\max_{\tDum \in [\tVal, \tFin]} \min\{\tariEval{\xSig}{\tDum}, \min_{\tDumDum \in [\tVal, \tDum]} \con(\xSig(\tDumDum), \tDumDum), \rob_{0,i}(\xSig, \tDum)\}\\
    & \le \max_{\tDum \in [\tVal, \tFin]} \min\{\tarStarEval{\xSig}{\tDum}, \min_{\tDumDum \in [\tVal, \tDum]} \con(\xSig(\tDumDum), \tDumDum)\}\\
    & < \rob(\xSig, \tVal).
\end{align*}
for all $i \in [N]$, which would contradict Lemma \ref{lemma:lattice-term-functions}\textit{(i)}. 
Since $S$ is also compact, it has a minimal element $\tDum^*$.

Observe that
\begin{align*}
    &\max_{\tDum \in [\tVal, \tFin]} \min\{\tar^*(\xSig(\tDum), \tDum), \min_{\tDumDum \in [\tVal, \tDum]} \con(\xSig(\tDumDum), \tDumDum), \rob(\xSig, \tDum)\} \\
    &\ge \min\{\tar^*(\xSig(\tDum^*), \tDum^*), \min_{\tDumDum \in [\tVal, \tDum^*]} \con(\xSig(\tDumDum), \tDumDum), \rob(\xSig, \tDum^*)\},
    %&=\lat_{i \in [N]} \min\{\tar^*(\xSig(\tDum^*), \tDum^*), \min_{\tDumDum \in [\tVal, \tDum^*]} \con(\xSig(\tDumDum), \tDumDum), \rob_i(\xSig, \tDum^*)\},
\end{align*}
so by Lemma \ref{lemma:lattice-term-functions}\textit(ii) it suffices to prove that $\rob_i(\xSig, \tDum^*) \ge \rob(\xSig, \tVal)$ for each $i \in \mathcal{I}$.
To reach a contradiction, suppose that there is some $j \in \mathcal{I}$ for which $\rob_j(\xSig, \tDum^*) < \rob(\xSig, \tVal) \le \rob_j(\xSig, \tVal)$.
Then there is some $\tDum \in [\tVal, \tDum^*)$ such that
\begin{align*}
    &\min\{\tar_j(\xSig(\tDum), \tDum), \min_{\tDumDum \in [\tVal, \tDum]} \con(\xSig(\tDumDum), \tDumDum), \rob_{0,j}(\xSig, \tDum)\} \ge \rob(\xSig, \tVal),
\end{align*}
contradicting the fact that $\tDum^*$ is the minimal element of $S$.\\

\noindent($\ge$) This direction follows directly from Step 1.\\

\noindent[Step 3]
We next observe that 
$$\val[\rob] \ge \min\{\tar_i, \con, \val[\lpf_i \circ (\rob_{0,i},\rob_{\lnot i})]\}$$
for each $i \in [N]$.
Indeed, for each $\xSig \in \xSigs$, $\tVal \in \RleT$, and $i \in [N]$, letting
$$c = \min\{\tariEval{\xSig}{\tVal}, \conEval{\xSig}{\tVal}\},$$ we have by Lemma \ref{lemma:lattice-term-functions}(iii) that
\begin{align*}
    \rob(\xSig, \tVal) &\ge \lpf_i(\min\{\tariEval{\xSig}{\tVal}, \conEval{\xSig}{\tVal}, \rob_{0,i}(\xSig, \tVal) \}, \rob_{\lnot i}(\xSig, \tVal))\\
    &=\lpf_i(\min\{c, \rob_{0,i}(\xSig, \tVal) \}, \rob_{\lnot i}(\xSig, \tVal))\\
    &\ge\lpf_i(\min\{c, \rob_{0,i}(\xSig, \tVal) \}, (\min\{c,\rob_{1}(\xSig, \tVal)\}, \dots,\\
    &\qquad~\min\{c,\rob_{i-1}(\xSig, \tVal)\},\min\{c,\rob_{i+1}(\xSig, \tVal)\}, \dots, \\
    &\qquad~\min\{c, \rob_N(\xSig, \tVal)\}))\\
    &= \min\{c,\lpf_i(\rob_{0,i}(\xSig, \tVal),\rob_{\lnot i}(\xSig, \tVal))\}\\
    &= \min\{\tariEval{\xVal}{\tVal}, \conEval{\xSig}{\tVal}, \lpf_i(\rob_{0,i}(\xSig, \tVal),\rob_{\lnot i}(\xSig, \tVal))\}.
\end{align*}
Thus for any $\xVal \in \Rn$, $\tVal \in \RleT$, and $i \in [N]$ we have
\begin{align*}
    \val[\rob](\xVal, \tVal) &\ge \infd \supu \min\{\tariEval{\stratTraj}{\tVal}, \conEval{\stratTraj}{\tVal}, \\
    &\qquad\qquad\qquad\quad\lpf_i(\rob_{0,i}(\stratTraj, \tVal),\rob_{\lnot i}(\stratTraj, \tVal))\} \\
    &= \min\{\tar_i(\xVal, \tVal), \con(\xVal, \tVal), \val[\lpf_i \circ (\rob_{0,i},\rob_{\lnot i})](\xVal, \tVal)\}.
\end{align*}

\noindent[Step 4]
In light of Step 2, we can apply Theorem \ref{thm:main-theorem} to obtain
\begin{align*}
    \val[\rob] &= \valRA[\min\{\tar^*, \val[\rob]\},\con;\tFin]\\
    &= \valRA[\min\{\tar^*, \con, \val[\rob]\},\con;\tFin].
\end{align*}
It thus suffices to prove that $\min\{\tarSpec, \con\} = \min\{\tar^*, \con, \val[\rob]\}$.
Let $\xVal \in \Rn$ and $\tVal \le \tFin$.
Choose $i \in \argmax_{j \in [N]} \tar_j(\xVal, \tVal)$.

Now observe that for each $\xSig \in \xSigs$, we have
\begin{align*}
    \rob_i(\xSig, \tVal) &\le \max_{\tDum \in [\tVal, \tFin]} \min\{ \min_{\tDumDum \in [\tVal, \tDum]} \conEval{\xSig}{\tDumDum}, \rob_{0,i}(\xSig, \tDum) \}\\
    &\le \rob_{0,i}(\xSig, \tVal)
\end{align*}
by the monotonicity assumption \eqref{eqn:montonicity-condition-i}.
Then for all $j \in [N]$,
\begin{align*}
&\val[\lpf_i \circ (\rob_{0,i}, \rob_{\lnot i})](\xVal, \tVal) \\
&\ge \val[\lpf_i \circ (\rob_{i}, \rob_{\lnot i})](\xVal, \tVal)\\
&= \val[\rob](\xVal, \tVal) \\
&\ge \min\{\tar_j(\xVal, \tVal), \con(\xVal, \tVal), \val[\lpf_j \circ (\rob_{0,j}, \rob_{\lnot j})](\xVal, \tVal)\},
\end{align*}
and thus
\begin{align*}
&\min\{\tar_i(\xVal, \tVal), \con(\xVal, \tVal), \val[\lpf_i \circ (\rob_{0,i}, \rob_{\lnot i})](\xVal, \tVal)\} \\
&\ge \min\{\tar_i(\xVal, \tVal), \con(\xVal, \tVal), \val[\rob](\xVal, \tVal)\} \\
&\ge \min\{\tar_j(\xVal, \tVal), \con(\xVal, \tVal), \val[\lpf_j \circ (\rob_{0,j}, \rob_{\lnot j})](\xVal, \tVal)\}.
\end{align*}
But then
\begin{align*}
&\min\{\tarSpec(\xVal, \tVal), \con(\xVal, \tVal)\} \\
&= \max_{j \in [N]} \min\{\tar_j(\xVal, \tVal), \con(\xVal, \tVal), \val[\lpf_j \circ (\rob_{0,j}, \rob_{\lnot j})](\xVal, \tVal)\},\\
&= \min\{\tar_i(\xVal, \tVal), \con(\xVal, \tVal), \val[\rob](\xVal, \tVal)\}\\
&= \min\{\tar^*(\xVal, \tVal), \con(\xVal, \tVal), \val[\rob](\xVal, \tVal)\}. \qedhere
\end{align*}
\end{proof}

\bibliographystyle{ieeetr}
\bibliography{references}

\end{document}